\documentclass[journal]{IEEEtran}
\usepackage{settings}

\begin{document}

\title{Logarithmic Depth Decomposition of Approximate Multi-Controlled Single-Qubit Gates Without Ancilla Qubits}

\author{
        Jefferson D. S. Silva\orcidlink{0000-0002-2063-6140},
        Adenilton J. da Silva\orcidlink{0000-0003-0019-7694}
        \thanks{
        This work is supported by research grants from CNPq, CAPES, and FACEPE (Brazilian research agencies).
        
        J.D.S. Silva and A.J. da Silva are with Centro de Informática, Universidade Federal de Pernambuco, Recife, Pernambuco, 50.740-560, Brazil (e-mail: jdss2@cin.ufpe.br; ajsilva@cin.ufpe.br).}}
\maketitle

\begin{abstract}
The synthesis of quantum operators involves decomposing general quantum gates into the gate set supported by a given quantum device. Multi-controlled gates are essential components in this process. In this work, we present an improved decomposition of multi-controlled NOT gates with logarithmic depth using a single ancilla qubit while reducing the ancillary resource requirements compared to previous work. We further introduce a relative-phase multi-controlled NOT gate that eliminates the need for ancillas. Building on these results, we optimize a previously proposed decomposition of multi-target, multi-controlled special unitary SU(2) gates by identifying the presence of a conditionally clean qubit. 
Additionally, we introduce the best-known decomposition of multi-controlled approximate unitary U(2) gates, which do not require ancilla qubits. This approach significantly reduces the overall circuit depth and CNOT count while preserving an adjustable error parameter, yielding a more efficient and scalable solution for synthesizing large controlled-unitary gates. Our method is particularly suitable for both NISQ and fault-tolerant quantum architectures. All software developed in this project is freely available.
\end{abstract}


\section{Introduction}

Quantum computing has emerged as a transformative technology capable of revolutionizing cryptography, optimization, and material science~\cite{shor1999polynomial, zhou2020quantum, cao2019quantum, bauer2020quantum, preskill2018quantum}. However, implementing quantum algorithms on near-term quantum devices, known as Noisy Intermediate-Scale Quantum (NISQ) devices~\cite{preskill2018quantum}, poses significant challenges due to limited qubit coherence times and high error rates. One of the critical tasks in quantum circuit design is the efficient decomposition of complex quantum gates into elementary gates, such as single-qubit gates and controlled-NOT (CNOT) gates, to minimize circuit depth and gate count. 

Multi-controlled gates play a central role in quantum circuit design and have motivated a wide range of studies on their decomposition and optimization~\cite{barenco_1995, shende2006synthesis, maslov2016advantages, iten2016quantum, he2017decompositions, adenilton2022linear, vale2023circuit, claudon2024polylogarithmic, silva2024linear, nie2024quantum, khattar2025rise}.
A prominent example is the $n$-qubit Toffoli gate~\cite{barenco_1995}, which applies an $X$ gate to a target qubit conditioned on the state of $n$ control qubits. The efficient implementation of such gates is crucial for a wide range of applications, including quantum error correction~\cite{inada2021measurement}, quantum arithmetic~\cite{remaud2025ancilla}, and quantum simulation~\cite{morgado2021quantum}. 

Recent advancements in quantum circuit optimization have introduced the concept of \textit{conditionally clean ancillae}, which are qubits that are in a known state conditioned on the state of other qubits~\cite{nie2024quantum}. This concept has been leveraged to reduce the gate count and depth of quantum circuits and the total number of ancillary qubits required in the circuit design, particularly in the decomposition of $n$-qubit Toffoli gates with $O(log(n))$ depth.

Ref.~\cite{khattar2025rise} achieves a decomposition of the $n$-qubit Toffoli gate with depth $\mathcal{O}(\log n)$. This decomposition has a CNOT count of $6n-6$ using two clean ancillae or $12n-18$ using two dirty ancilla qubits. 

This approach represents an improvement over previous methods, which either required a larger number of ancillary qubits or resulted in higher gate counts.

In parallel, research on approximate decompositions of multi-controlled single-qubit gates has shown promise in reducing the complexity of quantum circuits. In Ref.~\cite{silva2024linear}, the authors proposed a method for approximating $n$-controlled single-qubit gates with size and depth $\mathcal{O}(n)$ without ancilla qubits, and with size $\mathcal{O}(n)$ and depth $\mathcal{O}(\log{n})$ using $n - 3$ ancilla qubits, significantly reducing the CNOT gate count while maintaining a controlled approximation error. In Ref.~\cite{claudon2024polylogarithmic}, the authors introduced an approximate $n$-controlled single-qubit gate decomposition with an adjustable error parameter $\epsilon$, achieving a depth of $\mathcal{O}(\log{n^3} \cdot \log{(1/\epsilon)})$ and size $\mathcal{O}(n \log{n^4})$ without ancilla qubits. These approaches are particularly relevant for NISQ devices, where the trade-off between circuit complexity and approximation error can be carefully managed to enhance the overall performance of quantum algorithms.

In this work, we introduce a relative-phase $n$-qubit Toffoli gate that eliminates the need for ancillas. Additionally, we present an improved Toffoli gate decomposition based on~\cite{khattar2025rise}, reducing the ancillary count to one and allowing a recursive application of the decomposition. Building on these results, we develop $n$-controlled, $m$-target $X$ and $\mathrm{SU(2)}$ operators with $\mathcal{O}(\log{n} + \log{m})$ depth, enhancing the approximate decomposition of multi-controlled $\mathrm{U(2)}$ operators proposed by Silva et al.~\cite{silva2024linear}. Our construction achieves ancilla-free operation in the final circuit with $\mathcal{O}(\log n)$ depth and a lower CNOT cost.
 
The remainder of this paper is organized as follows. Section II reviews related work on multi-controlled gate decompositions and conditionally clean ancillae. Section III presents our ancilla-free $n$-qubit Toffoli gate with relative phase and logarithmic depth. Section IV develops the single-ancilla variant with identical asymptotic performance. Section V extends these results to multi-controlled single-qubit gates, including: exact $\mathrm{SU(2)}$ and approximate $\mathrm{U(2)}$ implementations, both maintaining logarithmic depth. Section VI concludes with future research directions.

\section{Related Works}

\subsection{2 and 3-qubit Relative-Phase Toffoli Gates}

In the context of quantum circuit decomposition, controlled gates, such as the CNOT, are fundamental for implementing complex operations. However, these controlled gates can introduce significantly larger errors compared to single-qubit gates \cite{iten2016quantum}, making the minimization of their number essential for high-fidelity circuit execution.

The Toffoli gate, a key component for higher-order controlled operations, can be decomposed efficiently. Lemma 6 of Ref.~\cite{iten2016quantum} provides one such decomposition, implementing the Toffoli gate up to a diagonal gate using single-qubit rotations and three CNOT gates. These results are further elaborated in Ref.~\cite{maslov2016advantages}, where the same decomposition is expressed in terms of $T$ gates, as illustrated in Fig.~\ref{fig:relative_phase}.

\begin{figure}[htbp]
\centerline{\includegraphics[width=0.9\linewidth]{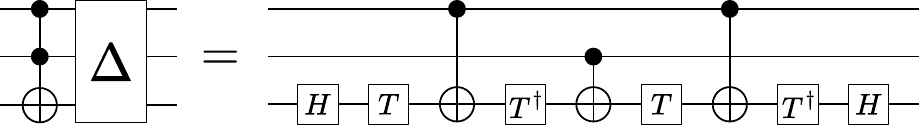}}
\caption{Implementation of a Toffoli gate up to a diagonal requiring only 3 CNOT gates. Adapted from~\cite{maslov2016advantages}.}
\label{fig:relative_phase}
\end{figure}

Furthermore, Ref.~\cite{maslov2016advantages} also proposes a decomposition for a 3-qubit Toffoli gate with a relative phase, which requires six CNOT gates and single-qubit rotations (see Fig.~\ref{fig:relative_phase3}).

\begin{figure}[htbp]
\centerline{\includegraphics[width=\linewidth]{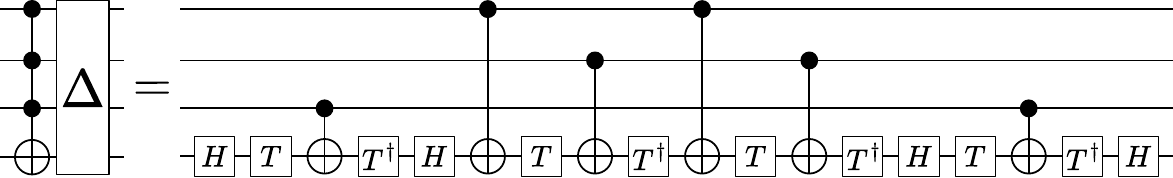}}
\caption{Implementation of a 3-qubit Toffoli gate up to a diagonal gate, requiring only 6 CNOT gates. Adapted from~\cite{maslov2016advantages}.}
\label{fig:relative_phase3}
\end{figure}

Both of these relative-phase constructions are particularly useful for reducing the number of CNOT gates in quantum circuits, a critical step towards mitigating errors.

\subsection{Multi-Controlled, Multi-target SU(2) decomposition}
 Ref.~\cite{silva2024linear} presents a decomposition of $n$-controlled, $m$-target $\mathrm{SU(2)}$ gates that achieve linear depth and CNOT cost. The corresponding circuit design is shown in Fig.~\ref{fig:mcmt_su2}. In this decomposition, the control qubits are divided into two groups: $k_1 = \lceil n/2 \rceil$ and $k_2 = \lfloor n/2 \rfloor$. The decomposition works as follows: when all control qubits are active, the circuit implements $W_i = \sigma_x \cdot A_i \cdot \sigma_x \cdot A_i^\dagger \cdot \sigma_x \cdot A_i \cdot \sigma_x \cdot A_i^\dagger$, where $A_i \in \mathrm{SU(2)}$. However, if any control qubit is inactive, the circuit implements the identity operation on all target qubits. The resulting circuit requires at most $16n + 8m - 32$ CNOT gates and has a depth of $32n + 8m - 52$.

\begin{figure}[htbp]
\centerline{\includegraphics[width=\linewidth]{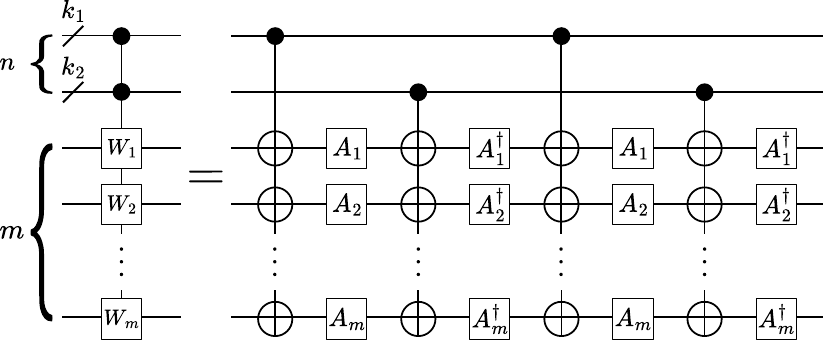}}
\caption{Circuit design scheme for $n$-controlled, $m$-target $\mathrm{SU(2)}$ gates~\cite{silva2024linear}, where each $W_i$ ($i = 1, \ldots, m$) has at least one real-valued diagonal element.}
\label{fig:mcmt_su2} 
\end{figure}

\subsection{Approximate Decomposition}

\begin{figure}[t]
\centerline{\includegraphics[width=\linewidth]{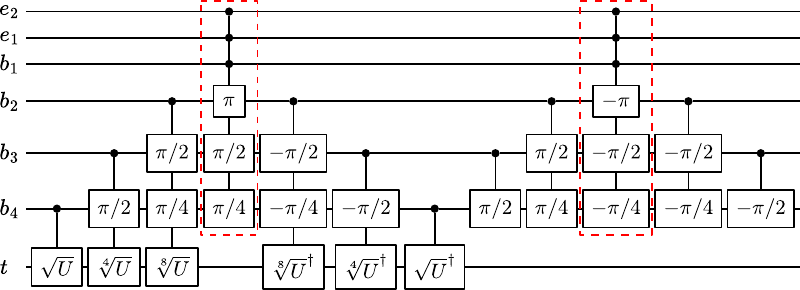}}
\caption{Approximate decomposition of a multi-controlled $\mathrm{U(2)}$ gate with six control qubits. The six control qubits are split into two groups: the extra control qubits $n_e=\{e_1, e_2\}$ and the base control qubits $n_b=\{b_1, b_2, b_3, b_4\}$. The blocks highlighted in red correspond to the decomposition of a multi-controlled, multi-target $\mathrm{SU(2)}$ operation. Adapted from the method described in~\cite{silva2024linear}. The $R_x$ gates are labeled with their corresponding rotation angles $\pm\pi/2^k$, for $k \in \mathbb{N}$.}
\label{fig:mcu_approx}
\end{figure}

The decomposition of approximate $n$-controlled $U$ gates, where $U \in \mathrm{U}(2)$, introduced in~\cite{silva2024linear}, follows a structured approach that employs controlled $R_x$ rotations to implement multi-controlled $U$ gates with linear CNOT count and depth, as illustrated in Fig.~\ref{fig:mcu_approx}.

In the exact decomposition design~\cite{adenilton2022linear}, when all control qubits are active, the circuit applies successive roots of $U$ to the target qubit, ensuring that their product results in $U$, and applies the identity otherwise. To obtain an approximate version, the lowest-index root is removed, introducing a bounded approximation error $\epsilon$.

The number of control qubits $n$ in the approximate circuit design is divided into two groups. The first group, denoted $n_b$, is determined by fixing the approximation error $\epsilon$ according to Eq.~\eqref{scale_nb}:
\begin{equation}\label{scale_nb}
    n_b \geq \log_2\left(\frac{|\theta|}{\arccos\!\left(1 - \epsilon^2 / 2\right)} \right),
\end{equation}
where $\theta$ depends on the original $U$ operator~\cite{silva2024linear}. Specifically, $n_b$ corresponds to the smallest root of the $U$ index, expressed as $U^{1/2^{n_b}-1}$, and therefore represents the number of control qubits to be decomposed into $R_x$ gates.

The second group, denoted $n_e$, is introduced to preserve the scalability of the circuit concerning the number of control qubits. Here, $n_b - 1$ represents the number of $R_x$ gates controlled by $n_e$. In this way, the circuit complexity is concentrated in its central stage, where an $(n_e + 1)$-controlled, $(n_b - 1)$-target $\mathrm{SU}(2)$ operation is employed, as illustrated in Fig.~\ref{fig:mcu_approx} for the case of six controls.

For a given target error $\epsilon$, the number of base control qubits $n_b$ remains fixed, while the decomposition of $(n_e + 1)$-controlled, $(n_b - 1)$-target $R_x$ gates is used to manage the growth in the number of control qubits, as described in~\cite[Theorem~3]{silva2024linear}. This results in a CNOT cost of $4{\left( n_b - 1\right)}^2 + 32n - 16n_b - 48$ and a depth of $64n - 8n_b - 136$, where $n = n_b + n_e$.

\subsection{Conditionally Clean Ancillae}

Ref.~\cite{nie2024quantum} introduced the concept of conditionally clean ancillae, which are qubits whose state is guaranteed to be on a known basis (e.g. $\ket{0}$) only when conditioned on specific control qubits. For example, if a clean qubit $a$ stores the logical AND of controls $\{c_1, c_2, \dots, c_k\}$, then $a = \ket{1}$ implies $\bigwedge_{i=1}^k c_i = 1$. This property allows $c_1, \dots, c_k$ to act as a temporary workspace in subsequent computations without additional ancillas, provided that operations are restricted to the subspace where $a = \ket{1}$. This technique reduces ancillary dependency while preserving reversibility, making it a powerful tool for optimizing multi-controlled gates. Ref.~\cite{nie2024quantum} proposes a decomposition of CNOT gates with $O(log(n))$ depth.  

\subsection{Log-Depth Decomposition of Toffoli Gates with Conditionally Clean Ancillae}
 
Based on~\cite{claudon2024polylogarithmic, nie2024quantum}, Ref.~\cite{khattar2025rise} proposes a decomposition for $C^n X$ (the $n$-qubit Toffoli gate)  
with $2n - 3$ Toffoli gates and $\mathcal{O}(\log n)$ depth using 2 clean ancillae, and $4n -8$ Toffoli gates with $\mathcal{O}(\log n)$ depth using 2 dirty ancillae. The approach recursively partitions controls into blocks, using intermediate results as conditionally clean ancillae to apply gates in parallel,  where intermediate controls are reused as a workspace, reducing both depth and the ancillary count compared to other methods~\cite{khattar2025rise}.  

Ref.~\cite [Section 5.2]{khattar2025rise} operates through five sequential steps. First, the protocol utilizes a Toffoli gate to initialize two conditionally clean qubits using a single clean ancilla. Second, Step 2 constructs a ladder of $n - \log n - 2$ Toffoli gates to accumulate the logical AND of the remaining controls into $\log n$ qubits. This step involves a down ladder of parallel Toffoli gates, spanning from the second to the last control qubit in $k$ layers, and an up ladder in $k-1$ layers, both arranged in a geometric progression. Considering the depth of a 2-qubit Toffoli gate as 1, this step incurs a total depth of $3 \log_2{\left( \frac{n}{2} +1\right)}$ layers.

Subsequently, the method applies their linear-depth Toffoli decomposition~\cite[Section 5.1]{khattar2025rise} to $\log{n}+1$ control qubits: the $\log{n}$ qubits storing the logical AND from the previous step, and the clean ancilla qubit used by step 1.

This linear circuit, for $n$ controls, contains $2n-3$ Toffoli gates and exhibits a depth of $2n-1$. Furthermore, it requires one additional dedicated ancilla qubit. This construction stores the logical AND of all controls in a single control qubit, then applies a 2-qubit Toffoli gate with controls on this single control qubit and the clean ancilla, acting on the target qubit. Steps 4 and 5 components of the logarithmic-depth construction uncompute steps 2 and 1, respectively. The overall circuit depth is:

\begin{equation}\label{khattar_depth_5.2}
    \mathrm{Depth: } \, 2\log_2{\left(n\right)} + 6 \log_2{\left( \frac{n}{2} +1\right)} + 1 = \mathcal{O} \left( \log{n} \right)
\end{equation}

For dirty ancilla implementations, the protocol extends this framework by appending steps 2, 3, and 4 to the end of the circuit, increasing the total Toffoli count from $2n-3$ to $4n-8$ while preserving the depth characteristic of $\mathcal{O}(\log n)$.

\section{Logarithmic-Depth n-Qubit Toffoli Gates with Relative Phase and No Ancillae}\label{relative_phase_sec}

In this section, we present a logarithmic-depth decomposition of $n$-qubit Toffoli gates with a relative phase. In our approach, no dedicated ancilla qubits are required—a critical advantage for near-term quantum hardware—making it particularly suitable for NISQ platforms and complex circuit designs.

The construction of the original $n$-qubit Toffoli gate described in~\cite{khattar2025rise} requires two clean ancilla qubits. 
Here we show that it is possible to synthetise relative phase Toffoli gates without ancilla qubits. The main idea is to use conditionally clean qubits instead of ancilla qubits.
The disadvantage is the addition of 10 CNOT gates in circuit decomposition compared to the $n$-qubit Toffoli gate described in~\cite{khattar2025rise}. 

As illustrated in Fig.~\ref{fig:mrcx} for the case $n = 12$. Our ancilla-free construction uses 3-qubit Toffoli gates with action in the target qubit instead of a dedicated ancilla qubit. The extra control in steps 1 and 5 is used to create a conditionally clean ancilla qubit that will be used in step~3, which is the linear decomposition of $\log{(n-1)}$-qubit Toffoli gates that requires an ancilla qubit. 

The circuit exhibits three distinct modes of operation: when steps 1 and 5 are inactive but step 3 is active, it implements $A^2 \cdot \sigma_x \cdot {A^\dagger}^2 = -\ket{0}\bra{0} + \ket{1}\bra{1}$; when step 3 is inactive, it reduces to identity; and with all controls active, the sequence $A \cdot \sigma_x \cdot A \cdot \sigma_x \cdot A^\dagger \cdot \sigma_x \cdot A^\dagger = \sigma_x$. These behaviors emerge because $A$ and $A^\dagger$ act unconditionally on the target, making the operational modes entirely dependent on the control logic of the Toffoli gates~\cite[Lemma~6]{iten2016quantum}.

\begin{figure}[htbp]
\centerline{\includegraphics[width=\linewidth]{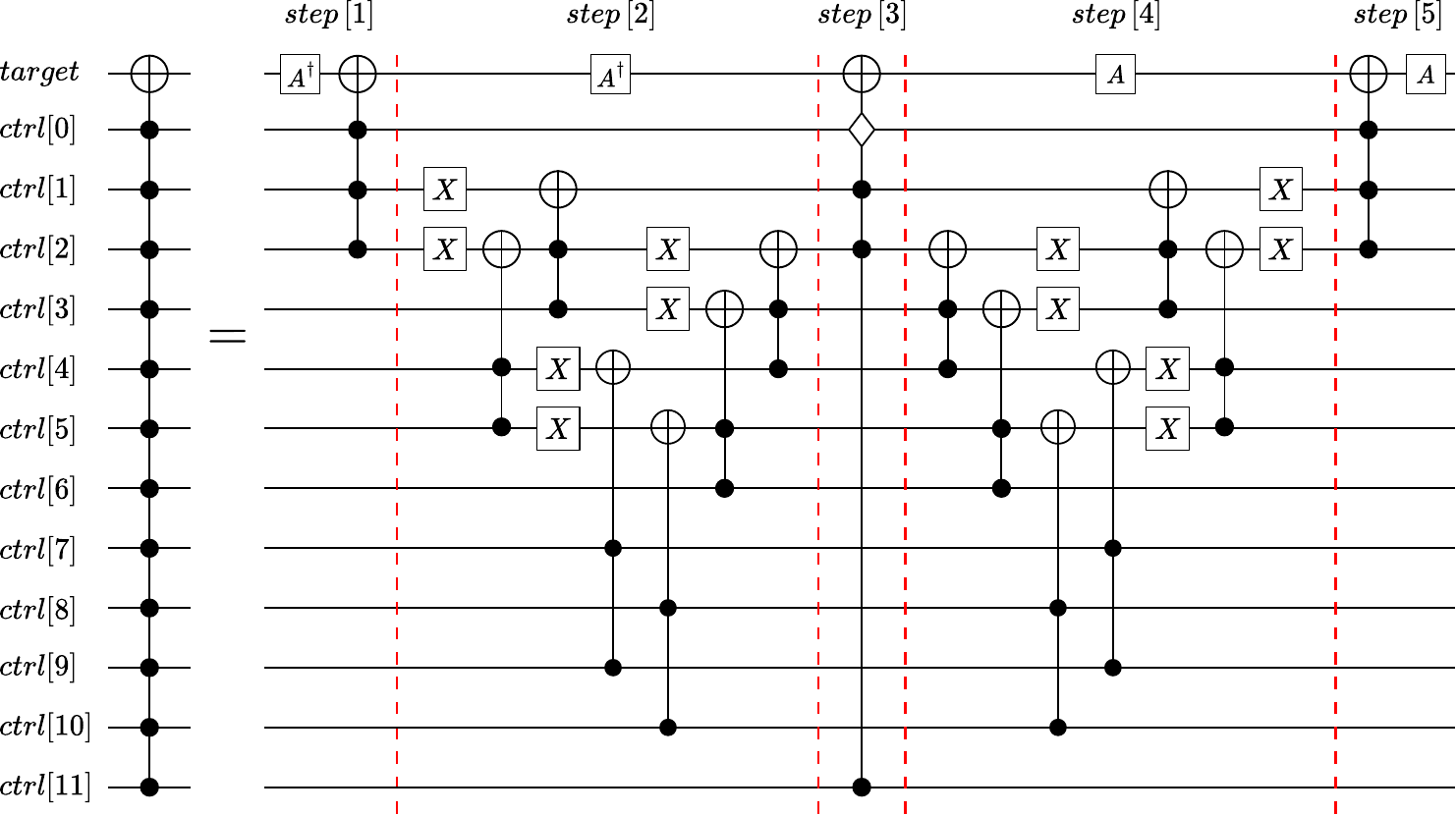}}
\caption{Decomposition of an $n$-qubit Toffoli gate with a relative phase and logarithmic depth \(\mathcal{O}(\log{n})\). Here, $A = R_y(\pi/4)$, and the white diamond on $ctrl[0]$ denotes a conditionally clean ancilla qubit used in Step~3. This ancilla enables the recursive application of the logarithmic-depth decomposition without requiring dedicated ancillas.}
\label{fig:mrcx}
\end{figure}

In this construction, Toffoli gates in steps 2 and 4 maintain identical designs and logical roles to those of~\cite{khattar2025rise}, but applied to the $n-1$ control qubits. For circuits with fewer than 11 control qubits, step 3 requires only a standard 2-qubit Toffoli gate, allowing steps 1 and 5 to use 2-qubit Toffoli gates instead of 3-qubit versions. This configuration eliminates the need for the third conditional ancilla. Theorem~\ref{teo:log_mrcx} formalizes this construction, proving that an $n$-qubit Toffoli gate with relative phase can achieve $\mathcal{O}(\log n)$ depth with at most $6n + 4$ CNOT gates.

\begin{theorem}\label{teo:log_mrcx}
An $n$-qubit Toffoli gate with a relative phase can be implemented using at most $6n + 4$ CNOT gates with $\mathcal{O}(\log n)$ depth.
\end{theorem}
    
\begin{proof}
The CNOT gate count in Steps~2 and~4 follows the construction in Ref.~\cite{khattar2025rise}: Step~2 requires $(n-1) - \log(n-1) - 2$ Toffoli gates, while Step~3 acts on $\log(n-1)$ control qubits (since the circuit has no dedicated ancilla qubits) and requires $2\log(n-1) - 3$ Toffoli gates. In Step~3, the target qubit requires a standard Toffoli gate, which can be implemented with 6 CNOT gates, while all other Toffoli gates in Steps~2 and~3 can be implemented using~\cite[Lemma~6]{iten2016quantum}, each requiring 3 CNOT gates. Furthermore, step~1 implements a 3-qubit Toffoli gate which, according to current Qiskit transpilation, requires 14 CNOT gates. steps~4 and~5 correspond to the inverse operations of Steps~2 and~1, respectively. Summing all contributions results in a total of $6n + 4$ CNOT gates.

For the depth analysis, Steps 1 and 5 each contribute a constant depth of 27, as verified after Qiskit transpilation. Steps~2 and~4 (being inverses) each contribute $3\log_2\!\left(\frac{n-1}{2} + 1\right)$, while Step~3 adds $2\log_2(n-1) - 1$, assuming that the depth of a single Toffoli gate is counted as~1. Since the logarithmic terms dominate the constant contributions, the overall depth is $\mathcal{O}(\log n)$.
\end{proof}

The circuit design utilizes control qubits in dual roles to reduce ancillary resource requirements. When Steps 1 and 5 are activated, the Toffoli operation guarantees all three control qubits are in state $\ket{1}$, allowing one control to serve as a conditionally clean ancilla for Step 3. Conversely, when Step~3 activates independently, the circuit's symmetry ensures the target qubit undergoes the relative-phase operation from~\cite[Lemma 6]{iten2016quantum}, maintaining functionality without dedicated ancillae. This dual-role approach achieves ancillary reduction while preserving the $\mathcal{O}(\log n)$ depth scaling.

\section{Logarithmic-Depth n-Qubit Toffoli Gates with one Ancillae}\label{sec3}

We present a modified circuit decomposition for $n$-qubit Toffoli gates that reduces the ancilla qubits requirements while maintaining logarithmic depth. Building on the logarithmic-depth construction from~\cite{khattar2025rise}, we extend the initial conditionally clean qubit creation from 2 to 3 qubits. 

This modification allows Step 3---the linear implementation on $\log{n} + 1$ qubits---to utilize one of these 3 control qubits of Step~1 as a conditionally clean ancilla qubit. This qubit should be one of the $n - \log{n}$ qubits not involved in storing the logical AND from Step 2, eliminating the need for dedicated ancillae in the recursive step.

To construct an $n$-qubit Toffoli gate with $\mathcal{O}(\log n)$ depth using one clean ancilla qubit, we implement the circuit through five steps with mirror symmetry: Step~4 inverts Step~2, and Step~5 inverts Step~1. Fig.~\ref{fig:mcx12} demonstrates this construction for $n=12$.

\begin{figure}[htbp]
\centerline{\includegraphics[width=\linewidth]{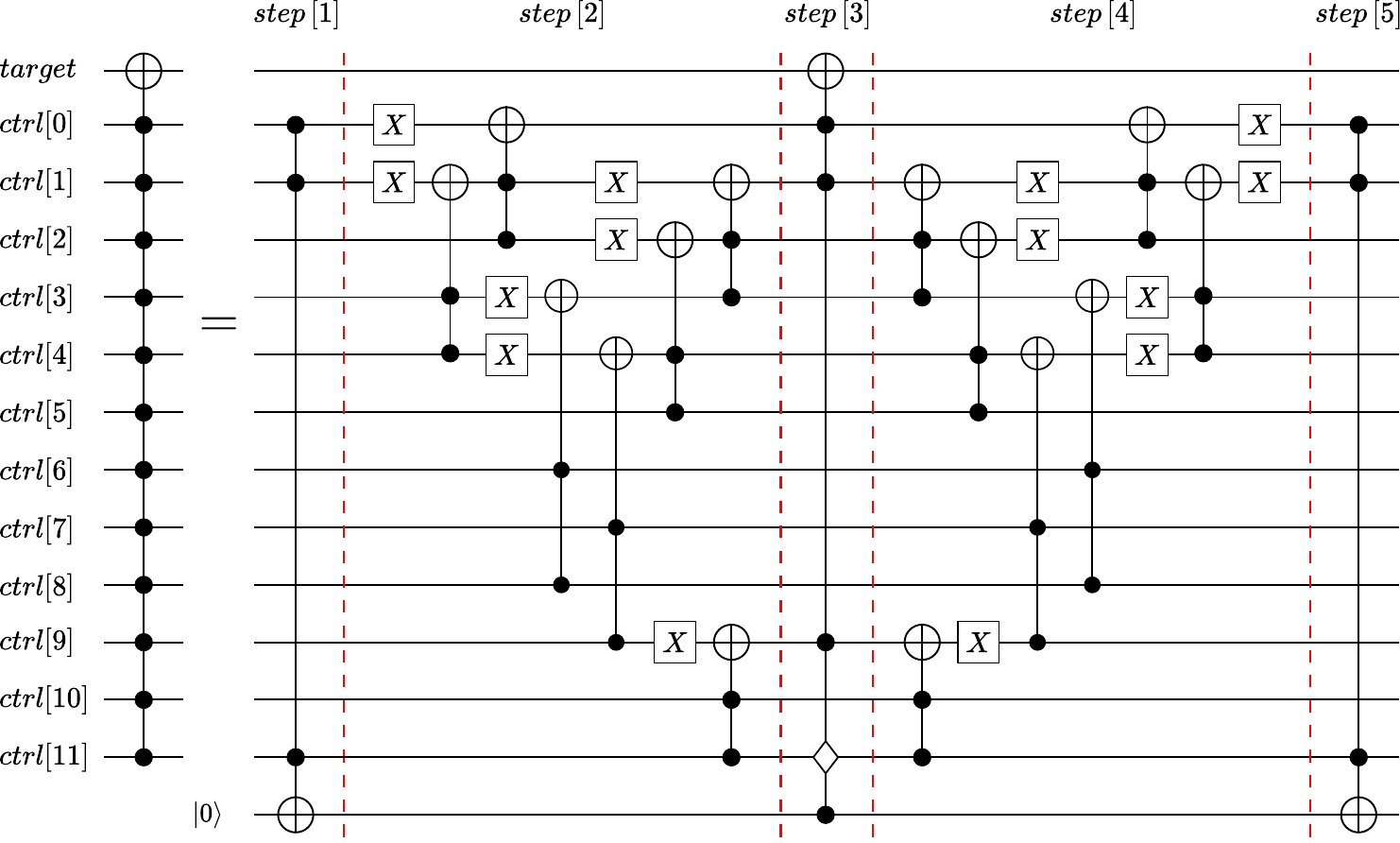}}
\caption{Implementation of an $n$-qubit Toffoli gate using one clean ancilla qubit, illustrated for $n=12$ with $\mathcal{O}(\log n)$ depth. The white diamond on qubit $\mathit{ctrl}[11]$ marks the conditionally clean ancilla used for Step~3, enabling the recursive logarithmic-depth decomposition.}
\label{fig:mcx12}
\end{figure}

Step 1 employs a 3-qubit Toffoli gate where two controls generate conditionally clean ancillas for Step 2, while the third becomes a conditionally clean ancilla for Step 3. The subsequent steps maintain this resource efficiency: Step 3 performs the linear decomposition from~\cite{khattar2025rise} on $\log{(n)} + 1$ qubits, while Steps 4 and 5 respectively uncompute Steps 2 and 1, completing the circuit without additional ancillas.

\begin{theorem}\label{teo:log_mcx}
    It is possible to construct a quantum circuit that implements an $n$-qubit Toffoli gate using one clean ancilla, with a total of $6n+2$~CNOT gates and depth $\mathcal{O}(\log{n})$.
\end{theorem}

\begin{proof}
The original decomposition in~\cite{khattar2025rise} requires $2n - 3$ Toffoli gates.  
In our modified construction, Steps 1 and 5 are replaced by 3-qubit Toffoli gates with relative phase, which, according to~\cite{maslov2016advantages}, require only $6$~CNOT gates each.  

Step~2 implements $n - \log{n} - 2$ Toffoli gates.
Step~3 applies the linear decomposition of~\cite{khattar2025rise} to $\log{n}$ qubits that store the logical AND computed in the previous step, but replaces the Toffoli gate acting on the target with a 3-qubit Toffoli gate. This gate uses the dedicated ancilla qubit as the third control, and treats the control qubit not used for storing the logical AND—namely, the third control introduced in Steps~1 and~5—as a conditionally clean ancilla for Step~3.
  
As a result, Step~3 accounts for $2\log{n} - 4$ Toffoli gates, plus one 3-qubit Toffoli gate requiring $14$~CNOT gates according to current Qiskit transpilation.

Summing all contributions, and noting that—except for the single 3-qubit Toffoli gate—all other Toffoli gates are implemented using the relative-phase construction (3~CNOTs each), we obtain:
\[
    2\cdot 6 \;+\; 3(2n - 8) \;+\; 14 \;=\; 6n + 2 \quad\text{CNOT gates.}
\]

For the depth analysis, Steps 2 and 4 retain the original structure from~\cite[Sec. 5.2]{khattar2025rise}, each contributing $3(\log(n/2) + 1)$ to the depth (assuming unit depth for standard Toffoli gates). Step~3 implements the linear decomposition on $\log n + 1$ qubits, where the target-acting Toffoli gate becomes a 3-qubit operation with constant depth. Steps 1 and 5 each contribute a depth of 18. Combining these terms, the total depth is $\mathcal{O}(\log n)$, where the asymptotic behavior is dominated by the logarithmic terms.
\end{proof}

In contrast to the relative-phase $n$-qubit Toffoli gate from Section~\ref{relative_phase_sec}, the current implementation modifies the base case of Step 3 to use a 3-qubit Toffoli gate. This modification ensures the target qubit only undergoes $\sigma_x$ when both Steps~1 and 5 are activated. Crucially, without this adjustment, Step~3 could execute when some control qubits are $\ket{0}$, leading to incorrect transformations. The ancilla qubit's additional control preserves the circuit's logical correctness while maintaining the $\mathcal{O}(\log n)$ depth complexity.

\section{Logarithmic-Depth Implementations of Multi-Controlled Single Qubit Gates}

Building upon our logarithmic depth decomposition of $n$-qubit Toffoli gates, we extend this approach to implement multi-controlled single-qubit operations. We present exact decompositions for $W \in \mathrm{SU(2)}$ gates, and approximate implementations for general $U \in \mathrm{U(2)}$ operators, both with $n$ control qubits. Each case maintains the logarithmic depth scaling while adapting the control structure through optimized ancilla management and gate synthesis. These results enable the efficient implementation of quantum algorithms that require controlled rotations or unitary transformations, particularly in scenarios where circuit depth is the limiting resource.

\subsection{Logarithmic-Depth Decomposition of Multi-Controlled, Multi-Target single qubit gates}

Multi-controlled multi-target gates arise in various quantum algorithms and state preparation routines~\cite{zindorf2024efficient, silva2024linear, tomesh2024quantum, gui2024spacetime, rosa2025optimizing}, particularly when the same unitary operation needs to be conditionally applied to multiple qubits. A straightforward adaptation of single-target decompositions to the multi-target case would require applying the same multi-controlled operation on each target sequentially, a cascade of decompositions.

In contrast, our multi-target approach leverages the fact that the control qubits are identical for every target, enabling a parallel implementation. Instead of executing a chain of multi-controlled operations, a single circuit is constructed that conditionally applies the desired unitary simultaneously to all target qubits when all control qubits are active. In this section, we propose a logarithmic-depth decomposition strategy for $n$-controlled $m$-target $W$ gates. Our approach employs an optimized implementation of a multi-controlled multi-target Toffoli gate and extends the linear-depth scheme of Ref.~\cite{silva2024linear} into a log-depth decomposition. As a result, the method achieves both a low CNOT gate count and reduced depth, making it particularly suitable for near-term quantum devices with limited coherence times.

We implement the $m$-target $X$ gate by applying a pair of CNOT gates to each target qubit, one before and one after the $n$-controlled $X$ operation, following the approach described in Ref.~\cite{wille2013improving}. As shown in Fig.~\ref{fig:mtx}, the circuit architecture consists of cascades of CNOT gates flanking an $n$-qubit Toffoli gate. A key optimization is achieved by rearranging the control qubits to group commuting CNOT operations into parallelizable layers, each containing at most $2^p$ gates. This design enables the construction of an $n$-controlled $m$-target $X$ gate ($C^n \bigotimes_{i=1}^m X_i$) with an optimal depth scaling of $\mathcal{O}(\log{n} + \log{m})$.

\begin{figure}[htbp]
\centerline{\includegraphics[width=\linewidth]{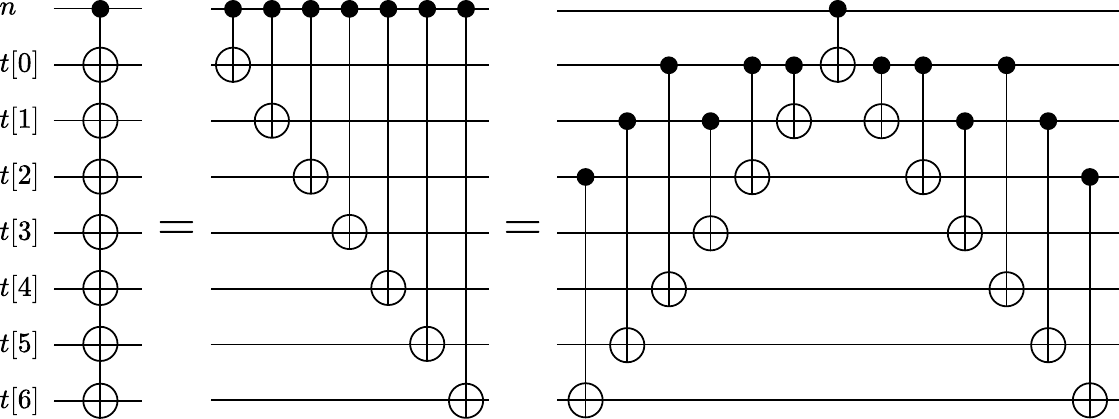}}
\caption{Log-depth multi-target $X$ gates scheme, illustrated for 7 targets. The circuit design features a cascade of CNOT gates surrounding an $n$-qubit Toffoli gate, with control qubits strategically arranged to enable parallelization.}
\label{fig:mtx}
\end{figure}

\begin{lemma}\label{lem:mcmtx}
    An $n$-controlled $m$-target Toffoli gate can be constructed using one clean ancilla qubit, with a total of $6n + 2m$ CNOT gates and a depth of $\mathcal{O}(\log{n} + \log{m})$.
\end{lemma}

\begin{proof}
    To implement the $n$-controlled part, we apply Theorem~\ref{teo:log_mcx}, which requires $6n +2$ CNOT gates and has a depth of $\mathcal{O}(\log{n})$. For the $m$-target part, we add 2 CNOT gates for each of $m-1$ additional targets. Thus, the total number of CNOT gates becomes:
    \[
        6n + 2m \,\mathrm{CNOTs}.
    \]
    To compute the circuit depth for the $m$-target part, we arrange commutative groups of CNOT gates in powers of two on both sides of the $n$-controlled Toffoli gate. Each side forms $L = \log(m)$ layers, resulting in a total of $2L$ layers. Therefore, the total depth is:
    \[
        \mathrm{Depth:} \, \mathcal{O}(\log{n} + \log{m}).
    \]
\end{proof}

Regarding the depth in the multi-target part, this decomposition remains particularly advantageous, even with a single control qubit. Although a single control qubit only requires a cascade of CNOT gates, effectively halving the CNOT count in the target part, the resulting depth is linear. In contrast, for $m$-target CNOT gates, Lemma~\ref{lem:mcmtx} achieves a depth of $\mathcal{O}(\log{m})$.

\begin{figure}[htbp]
\centerline{\includegraphics[width=\linewidth]{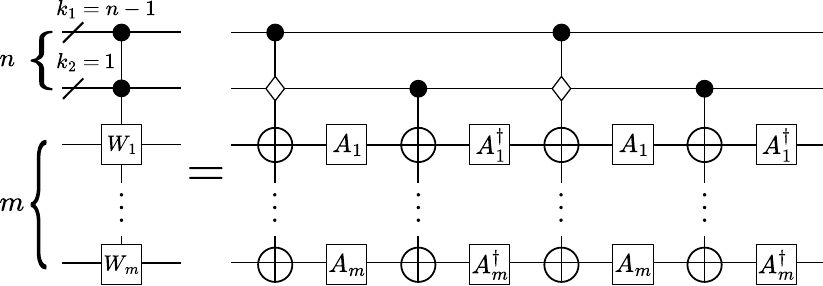}}
\caption{Logarithmic-depth decomposition of an $n$-controlled $m$-target $\mathrm{SU(2)}$ gate. The white diamond in $k_2$ is used as a conditionally clean ancilla qubit for the $k_1$-qubit $m$-target Toffoli gate. The design allows parallelization of gate blocks to minimize depth, using $k_2$ as a conditionally clean ancilla qubit.}
\label{fig:log_mcmt_su2}
\end{figure}

For $n$-controlled $m$-target $W$ decomposition, we retain the same circuit design as in Fig.~\ref{fig:mcmt_su2}. However, a few modifications are sufficient to achieve logarithmic depth. We can do $k_1 = n-1$ and $k_2 = 1$ and replace the $\{k_1, k_2\}$-qubit Toffoli gate in the circuit of Fig.~\ref{fig:mcmt_su2} by substituting Lemma 8 of Ref.~\cite{iten2016quantum} with our optimized version, described in Lemma~\ref{lem:mcmtx}, which reduces both the CNOT cost and depth. In this configuration, the qubit in $k_2$ is used as a conditionally clean ancilla by the $k_1$-qubit Toffoli gates, as shown in Fig.~\ref{fig:log_mcmt_su2}. This modification leads to the following theorem:

\begin{theorem}\label{teo:mcmtsu2}
    It is possible to construct a quantum circuit that implements an $n$-controlled, $m$-target $\mathrm{SU(2)}$ gate, $C^n \bigotimes_{i=1}^m W_i$, where $W_i \in \mathrm{SU(2)}$, using at most $12n + 8m - 14$ CNOT gates, with depth $\mathcal{O}(\log n + \log m)$.
\end{theorem}

\begin{proof}
    Consider an $n$-controlled, $m$-target $\mathrm{SU(2)}$ gate, $C^n \bigotimes_{i=1}^m W_i$, where $W_i \in \mathrm{SU(2)}$. We partition the $n$ control qubits into two groups: $k_1 = n - 1$ and $k_2 = 1$.  
    As shown in Fig.~\ref{fig:log_mcmt_su2}, the construction consists of eight steps: two $C^{k_1} \bigotimes_{i=1}^m X_i$ gates, two $C^{k_2} \bigotimes_{i=1}^m X_i$ gates, and four layers of $A$ and $A^\dagger$ gates.  

    For each $k_1$-controlled, $m$-target Toffoli gate, Lemma~\ref{lem:mcmtx} gives a cost of $6k_1 + 2m$ CNOTs, resulting in $2[6k_1 + 2m]$ for both gates.  
    For each $k_2$-controlled, $m$-target Toffoli gate, the first target is simply a CNOT (since $k_2 = 1$), giving a cost of $2(m-1) + 1$ CNOTs, or $2[2(m - 1) + 1]$ for both.  
    Summing all contributions yields:
    \[
        12n + 8m -  14\, \textrm{CNOTs}
    \]

    For the depth, the $A_i$ and $A_i^\dagger$ gates ($i = 1, \dots, m$) act in parallel, contributing with depth of 4. Using Lemma~\ref{lem:mcmtx}, the two $k_1$-controlled, $m$-target Toffoli gates, contributes with depth of $\mathcal{O}(\log{n}+ \log{m})$ each, and the two $k_2$-controlled, $m$-target Toffoli gate contributes with depth of $\mathcal{O}(\log{m})$, each. This way, the overall depth scales as $\mathcal{O}(\log{n} + \log{m})$, where the logarithmic terms dominate the asymptotic behavior.
\end{proof}

\begin{figure}[htbp]
\centerline{\includegraphics[width=\linewidth]{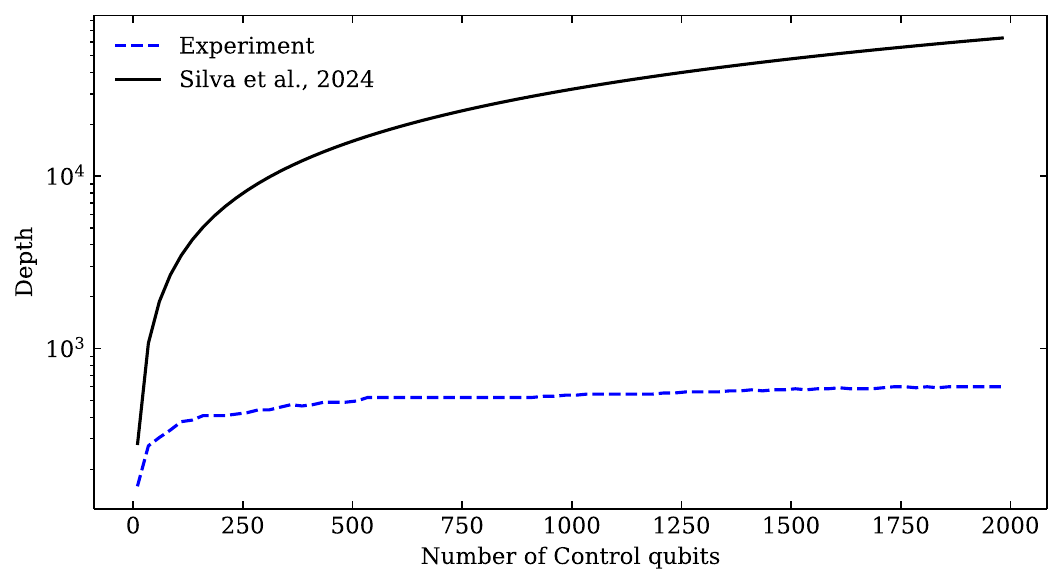}}
\caption{Comparative circuit depth (log scale) for decomposing multi-controlled $W \in \mathrm{SU(2)}$ gates. The black solid line shows the depth obtained after transpilation using Qiskit for the decomposition method predicted by~\cite[Theorem 2]{silva2024linear}. In contrast, the blue dashed line represents the depth resulting from our decomposition, also after Qiskit transpilation, and predicted by Theorem~\ref{teo:mcmtsu2}.}
\label{fig:mcmtsu2}
\end{figure}

Fig.~\ref{fig:mcmtsu2} shows a comparative analysis of circuit depth (in logarithmic scale) for the decomposition of multi-controlled $W \in\mathrm{ SU(2)}$ gates. The black line represents the depth of circuits generated using the decomposition method proposed by~\cite[Theorem 2]{silva2024linear}. In contrast, the blue dashed line represents the depth obtained from our proposed decomposition, both before and after transpilation using Qiskit. 

As shown in Fig.~\ref{fig:mcmtsu2}, our method achieves significantly lower circuit depth as the number of control qubits increases, which indicates that it is more efficient in practice. This improvement makes our decomposition a promising approach for quantum circuit optimization, especially in the context of NISQ devices, where circuit depth directly impacts fidelity due to noise, and in fault-tolerant architectures, where minimizing resources is crucial.

\subsection{Logarithmic-Depth Decomposition of Approximate Multi-Controlled U(2) gates}

 Theorem 3 of Ref.~\cite{silva2024linear} uses a linear depth multi-controlled, multi-target $\mathrm{SU(2)}$ decomposition to optimize their approximate multi-controlled $\mathrm{U(2)}$ decomposition. We propose replacing this multi-controlled, multi-target $\mathrm{SU(2)}$ decomposition with our newly optimized version, which features logarithmic depth and lower CNOT cost. This modification leads to the following theorem:

\begin{theorem}\label{teo: last}
    An approximate decomposition of a $n$-controlled $U$ gate ($C^n U$), where $U \in \mathrm{U}(2)$, can be implemented up to an error $\epsilon$ with a CNOT count upper bound of $4(n_b - 1)^2 + 24n - 8n_b -20$ and depth $\mathcal{O}(\log{n})$, where $n_b$ is a fixed base number of control qubits that depends on $\epsilon$ as in Eq.~\eqref{scale_nb}, and $n \geq n_b + 5$.
\end{theorem}

\begin{proof}
    The proof follows a similar structure to Theorem 3 of Ref.~\cite{silva2024linear}, but with a key difference: instead of using Theorem 2 of Ref.~\cite{silva2024linear} to account for $(n_e + 1)$-controlled $(n_b - 1)$-target $\mathrm{SU}(2)$ gates, we use our improved version described in Theorem~\ref{teo:mcmtsu2}. 

    By applying Theorem~\ref{teo:mcmtsu2}, we account for $2[12(n_e + 1) + 8(n_b - 1) -14]$ CNOTs. Adding these to the $4(n_b-1)^2$ CNOTs from the multi-controlled $U^{1/2^j}$ gates and expressing the result in terms of the total number of controls $n$ and $n_b$ (where $n = n_b + n_e$), the total CNOT count becomes:
    \[
        4(n_b - 1)^2 + 24n - 8n_b -20 \quad \text{CNOTs}.
    \]

    For the depth, note that the growth of the circuit depends on the extra control qubits $n_e$, while $n_b$ remains fixed up to an error $\epsilon$. Thus, asymptotically, the depth matches that of Theorem~\ref{teo:mcmtsu2}, i.e.: 
    $\mathcal{O}(\log{n})$ depth.
\end{proof}

\begin{figure}[htbp]
    \centerline{\includegraphics[width=\linewidth]{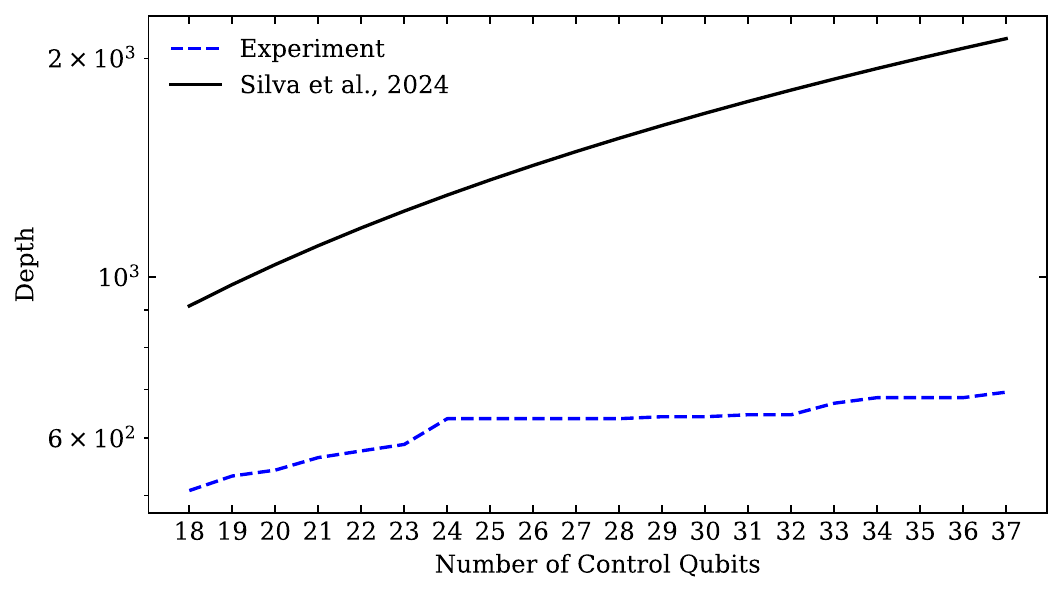}}
    \caption{
        Comparative circuit depth (logarithmic scale) for approximate decomposition of multi-controlled $U = X$ gates. 
        The blue dashed line represents our experimental results post-Qiskit transpilation (predicted by Theorem~\ref{teo: last}), 
        while the black solid line denotes the theoretical circuit depth from \cite[Theorem 3]{silva2024linear}.
    }
    \label{fig:mcu2}
\end{figure}

The approximation error bound in Eq.~\eqref{scale_nb} remains unchanged, as it depends solely on the base parameter $n_b$, preserved in our construction. By retaining the exact unitary transformation on each target qubit (per the original decomposition in~\cite{silva2024linear}), our substitution of the $(n_e+1)$-controlled $(n_b-1)$-target $R_x$ gates with Theorem~\ref{teo:mcmtsu2}'s optimized implementation modifies only the circuit resources—not the approximation's mathematical structure. Thus, while significantly improving efficiency, our approach preserves both the error scaling $\epsilon$ and all convergence guarantees.

Fig.~\ref{fig:mcu2} compares the circuit depth (logarithmic $Y$-axis) for approximate decomposition of multi-controlled $U = X$ gates. The blue dashed line shows our experimental results after Qiskit transpilation, aligned with Theorem~\ref{teo: last}, while the black solid line reflects the theoretical bounds from \cite[Theorem 3]{silva2024linear}. Our decomposition achieves logarithmic depth, critical for near-term quantum devices with limited coherence times. Both methods scale linearly in CNOT cost with the number of controls $n$, but our approach reduces the cost to $\mathcal{O}(24n)$, a $25\%$ improvement over Silva et al.'s $\mathcal{O}(32n)$.

This reduction in depth and CNOT count is crucial for NISQ devices, where gate errors and decoherence limit the depth of executable circuits. Logarithmic-depth decompositions can reduce error accumulation, potentially enhancing output fidelity. The lower CNOT overhead further reduces runtime and improves fidelity, as two-qubit gates dominate error rates in most architectures. 

\section{Conclusion}
\label{sec:conclusion}

This work presents comprehensive advances in the decomposition of multi-controlled quantum gates, achieving simultaneous improvements in depth complexity, ancilla requirements, and gate counts. We develop two fundamental implementations: an ancilla-free $n$-qubit Toffoli gate with relative phase and $\mathcal{O}(\log n)$ depth and a single clean-ancilla version with minimal overhead. These constructions build upon Ref.~\cite{khattar2025rise} through the application of conditionally clean qubits, with both the relative-phase and clean-ancilla versions introducing only minimal overhead compared to the original decomposition while reducing ancillary requirements.

The clean-ancilla variant cuts the original ancilla count in half while adding only 8 CNOT gates. Our techniques further enable $n$-controlled $m$-target operations with $\mathcal{O}(\log n + \log m)$ depth, substantially advancing multi-controlled single-qubit gate synthesis. Beyond Toffoli gates, we demonstrate exact $\mathrm{SU}(2)$ and approximate $\mathrm{U}(2)$ controlled operations with $\mathcal{O}(12n)$ and $\mathcal{O}(24n)$ CNOT counts, respectively, both achieving logarithmic depth scaling—a significant improvement over prior linear-depth approaches.

The conditionally clean qubit methodology proves particularly powerful, enabling relative-phase implementations that eliminate auxiliary qubits without compromising depth complexity. These advances offer concrete benefits for both near-term devices and fault-tolerant architectures, establishing new benchmarks for resource-efficient quantum circuit design. Future work may explore additional optimizations through gate cancellation techniques and shared ancilla protocols, potentially yielding further improvements for practical quantum algorithms relying on controlled operations.

\section*{Conflict of interest}
The authors declare no conflicts of interest.

\bibliographystyle{IEEEtran}
\bibliography{refs}
\end{document}